\documentclass{rspublic}
\usepackage{amssymb,amsmath}
\usepackage{rotating}
\usepackage{multirow}
\usepackage{graphicx}
\usepackage{array}
\usepackage{lscape}
\usepackage{amsthm}
\usepackage{enumerate}
\newtheorem{definition}{Definition}
\newtheorem{cor}[theorem]{Corollary}
\numberwithin{equation}{section}
\renewcommand*\thesection{\arabic{section}}

\begin{document}

\title[]{Interplay of symmetries and other integrability quantifiers in finite dimensional integrable nonlinear dynamical systems}

\author[Mohanasubha, Chandrasekar, Senthilvelan and Lakshmanan]{R. Mohanasubha$^1$, V. K. Chandrasekar$^2$, M. Senthilvelan$^1$ and M. Lakshmanan$^1$}

\affiliation{$^1$ Centre for Nonlinear Dynamics, School of Physics,
Bharathidasan University, Tiruchirappalli - 620 024, India\\
$^2$ Centre for Nonlinear Science and Engineering, School of Electrical and Electronics Engineering, SASTRA University, Thanjavur - 613 401, India}

\label{firstpage}

\maketitle

\begin{abstract}{Null forms, Symmetries, Darboux polynomials, Integrating factors and Jacobi last multiplier}
In this work, we establish a connection between the extended Prelle-Singer procedure with other widely used analytical methods to identify integrable systems in the case of $n^{th}$-order nonlinear ordinary differential equations (ODEs). By synthesizing these methods we bring out the interlink between Lie point symmetries, contact symmetries, $\lambda$-symmetries, adjoint-symmetries, null forms, Darboux polynomials, integrating factors, Jacobi last multiplier and generalized $\lambda$-symmetries corresponding to the $n^{th}$-order ODEs. We also prove these interlinks with suitable examples. By exploiting these interconnections, the characteristic quantities associated with different methods can be deduced without solving the associated determining equations. 
\end{abstract}

\section{Introduction}
In two of our earlier works \cite{suba1,suba2}, we have interconnected six widely used analytical methods for solving ordinary differential equations (ODEs), namely extended Prelle-Singer method, Lie symmetry analysis, Jacobi last multiplier method, Darboux method, adjoint-symmetries method and $\lambda$-symmetries approach by considering second- and third-order ODEs. Progressing further, in this paper, we unearth the interconnection between the extended Prelle-Singer procedure with two more procedures, namely (i) contact symmetries and (ii) generalized $\lambda$-symmetries. By establishing these new interconnections we are able to bring the eight different analytical methods under one umbrella. More importantly we prove the presence of the interconnections upto general $n^{th}$-order ODEs. Some of the results we report here are new to the literature as well. We also note here that the first attempt to connect the Prelle-Singer procedure with other methods, in particular $\lambda$-symmetries in second-order ODEs, came from Muriel and Romero \cite{Mur3}. They have shown that $\lambda$-symmetries are nothing but the null forms (with a negative sign) given in Prelle-Singer procedure. In this work, we study this interconnection between the null forms in the extended Prelle-Singer procedure and $\lambda$-symmetries to higher-orders. The other important result which we report in this paper is the interconnection between generalized $\lambda$-symmetries and the extended Prelle-Singer method quantities, namely null forms and integrating factor. It is very difficult to determine both of them by solving their determining equations. However using the proposed interconnections one can obtain these symmetries in a simple and straightforward manner. 

We demonstrate all these interconnections with suitable examples in all orders. For example, in second- and third- order ODEs, we recall the same examples which we considered earlier and demonstrate the new interconnections so that the wider interconnection can now be appreciated. To demonstrate the validity of the interconnections in higher-order we consider an example from fourth-order as well. The realization of the interplay at the $n^{th}$-order is also demonstrated through the $n^{th}$-order free particle equation.

Some of the interconnections are not mere extensions of the interconnections from the case of second and third-order ODEs. The interconnections is achieved by introducing suitable transformations in the Prelle-Singer procedure quantities, namely the null forms, $S_{i},~~i=1,2,\cdots,n-1$ and integrating factor, $R$. These interconnections will be briefly explained in the following sections.

The plan of the paper is as follows. In section 2, we give the definition and the determining equations of the Prelle-Singer procedure for solving $n^{th}$-order ODEs. In addition, we discuss the several well known methods such as Lie point symmetries, contact symmetries, $\lambda$-symmetries, generalized $\lambda$-symmetries, Jacobi last multiplier, Darboux polynomials and adjoint-symmetries methods and their determining equations. We also discuss the known and unknown interconnections among all these methods. In section 3, we show that interconnections among all these methods with examples. Finally we summarize our results in section 4.
\section{Interconnections}
To prove the interconnections among all the above said methods, we start our discussion with the extended Prelle-Singer method. In the following, we discuss the  all other methods in relation to the extended Prelle-Singer method.

Consider an $n^{th}$-order ODE of the following form
 \begin{equation} 
x^{(n)}=\phi(t,x,x^{(1)},x^{(2)},...,x^{(n-1)}),~~x^{(j)}=\frac{d^j x}{dt^j},~j=1,2,\cdots,n, \label{main1}
\end{equation}
where $\phi$ is the function of $t,~x,~x^{(1)},\cdots,x^{(n-1)}$.
\subsection{Extended Prelle-Singer method }
\label{sec2} 
In 1983, Prelle and Singer have proposed a procedure for solving first-order ODEs \cite{Prelle} that presents the solution, if such a solution exists, in terms of elementary functions. Subsequently Duarte et. al. extended the underlying ideas to second-order ODEs and constructed only one integral for a class of equations \cite{duart}. Later three of the present authors have extended the algorithm given by Duarte et. al. in such a way that the extended algorithm will provide two independent integrals for the given second-order ODE \cite{anna_2nd}. The same authors have extended the algorithm to third-order and $n^{th}$-order as well as coupled ODEs and established that the extended Prelle-singer procedure is a stand alone method to determine integrating factors, integrals and the general solution of ODEs of any order including coupled ones, provided it is integrable. For more details about this method, one may refer to Refs.\cite{anna_2nd,anna_3rd,anna_n}.

In the following, we recall briefly the extended Prelle-Singer procedure applicable for $n^{th}$-order ODEs.

\begin{definition}
Consider an $n^{th}$-order ODE (\ref{main1}). Let the ODE (\ref{main1}) admits a first integral $I(t,x,x^{(1)},...,x^{(n-1)})=C$, where $C$ is a constant. Let $S_i,i=1,2,...,n-1,$ and $R$ denote the null forms (which are essentially functions) and the integrating factor, respectively. They can then be determined from the relations \cite{anna_n}
\begin{subequations}
\begin{eqnarray}
D[S_1] & = -\phi_x+S_1\phi_{x^{(n-1)}}+S_1S_{n-1},
\qquad \qquad \qquad \qquad \qquad \qquad\;\; \label{net9}\\
D[S_i] & = -\phi_{x^{(i-1)}}+S_i\phi_{x^{(n-1)}}+S_iS_{i+1}-S_{i-1},
\;i=2,3,\dots,n-2,\label{net10}\\
D[S_{n-1}]&=-\phi_{x^{(n-2)}}+S_{n-1}\phi_{x^{(n-1)}}+S_{n-1}^2-S_{n-2} 
, \qquad \qquad \qquad \quad \;\;\label{net11}\\
D[R]  &= -R(S_{n-1}+\phi_{x^{(n-1)}}),
\qquad \qquad \qquad \qquad \qquad \qquad \qquad \quad\;\;\label{net12}\\
R_{x^{(i)}}S_1 & = -R{S_{1}}_{x^i}+R_{x}S_{i+1}+R{S_{i+1}}_{x}, 
\;\;i=1,2,\dots,n-2,  \quad \qquad\;  \label{net13}\\
R_{x^{(i)}}S_{j+1} & = -R{S_{j+1}}_{x^{(i)}}+R_{x^j}S_{i+1}+R{S_{i+1}}_{x^j},
\;\;i,j=1,2,\dots,n-2,\;\; \label{net14}\\
R_{x^{(i)}} & = R_{x^{(n-1)}}S_{i+1}+R{S_{i+1}}_{x^{(n-1)}},
\;\;i=1,2,\dots,n-2, \qquad \qquad \;\;\label{net15}\\
R_x & = R_{x^{(n-1)}}S_1+R{S_{1}}_{x^{(n-1)}},
\qquad \qquad \qquad \qquad \qquad \qquad \qquad \quad \label{net16}
\end{eqnarray}
\end{subequations}
where the total differential operator $D$ is defined by 
$D=\frac{\partial}{\partial{t}}+\sum_{i=1}^{n}x^{(i)}\frac{\partial}{\partial{x^{(i-1)}}}$. In the above, the null forms are functions such that $S_i(t,x,x^{(1)},\cdots,x^{(n-1)})x^{(i)}dt-S_i(t,x,x^{(1)},\cdots,x^{(n-1)}) dx^{(i-1)}=0,~i=1,2,\cdots n-1$.
\end{definition}
Once we know the null forms $S_i, i=1,2,\cdots n-1,$ and the integrating factor $R$, we can construct the integrals of the ODE (\ref{main1}) through the expression \cite{anna_n}
\begin{equation}
I(t,x,x^{(1)},...,x^{(n-1)}) = \sum_{i=1}^{n}r_i-\int\left[R+\frac{d}{dx^{(n-1)}} 
\left(\sum_{i=1}^{n}r_i
   \right)\right]dx^{(n-1)},
  \label{net17}
\end{equation}
where 
\begin{eqnarray} 
r_1 & = \int R\bigg(\phi+\sum_{i=1}^{n-1}S_i x^{(i)}\bigg)dt,
\qquad r_2 =-\int \bigg(RS_1+\frac{d}{dx}r_1\bigg) dx,\nonumber\\
r_j & =-\int \left[RS_{j-1}+\frac{d}{dx^{(j-1)}}
\left(\sum_{k=1}^{j-1}r_k\right)\right]dx^{(j-1)},\;\;j=3,\ldots,n. \nonumber
\end{eqnarray}
We note here that $I_{t}  = R(\phi+\sum_{i=1}^{n-1}S_i x^{(i)}), ~I_{x}  = -RS_1,
~I_{x^{(i)}}  = -RS_{i+1},\;\;i=1,2,\ldots,n-2, ~I_{x^{(n-1)}} = -R$. Once we know $n$ independent integrals, we can derive the general solution of the given $n^{th}$-order ODE from these integrals.


\subsubsection{Transformations}
In the extended Prelle-Singer method, the null forms $S_i,i=1,2,\cdots,n-1,$ and integrating factor $R$ play a major role in finding the integrals.  

To connect the null forms and integrating factors with other integrability quantifiers, we introduce the following transformations in $S_i$ and $R$, that is
\begin{subequations}
\label{s_tran_eqn}
\begin{eqnarray}
&&S_{n-1}=-\frac{D[V]}{V},~~n \geq 2,\label{sn1eq}\\
&&S_i=\frac{X_i}{V},~~i=1,2,3,...,n-2,~~n>2.\label{ueq}
\end{eqnarray}
\end{subequations}
where $V(t,x,...,x^{(n-1)})$ and $X_i(t,x,...,x^{(n-1)})$ are $n-1$ unknown functions and $D$ is the total differential operator. With this substitution, Eqs.(\ref{net11}), (\ref{net9}) and (\ref{net10}) now become a system of linear equations:
\begin{subequations}
\begin{eqnarray}
D^2[V]&=&D[V]\phi_{x^{(n-1)}}+\phi_{x^{(n-2)}}V+X_{n-2},~~n \geq 2,\label{veq}\\
D[X_1]&=&\phi_{x^{(n-1)}} X_1-\phi_x V,~~~~n>2\label{met141}\\
D[X_i]&=&\phi_{x^{(n-1)}} X_i-\phi_{x^{i-1}} V-X_{i-1},~~i=2,3,...,n-2,~~n>2.\label{met14}
\end{eqnarray}
\end{subequations}
We introduce yet another transformation,
\begin{equation}
R=\frac{V}{F},\label{rvf}
\end{equation}
where $F(t,x,x^{(1)},...,x^{(n-1)})$ is a function to be determined, in (\ref{net12}) so that the latter equation can be rewritten in a compact form in the new variable $F$ as
\begin{equation}  
D[F] = \phi_{x^{(n-1)}}F.\label{met15}
\end{equation}
Solving (\ref{met15}) we can obtain the explicit form of $F$. The transformations given in Eqs.(\ref{s_tran_eqn}) and (\ref{rvf}) interlink the null forms and integrating factor with other integrability quantifiers, say $\lambda$-symmetries, Lie point symmetries, contact symmetries, Darboux polynomials, Jacobi last multiplier, adjoint-symmetries  and generalized $\lambda$-symmetries as we see below. 

\subsection{Lie point symmetry analysis}
Lie point symmetry analysis is one of the powerful methods to derive classes of solutions of differential equations in an algorithmic way. The underlying idea here is to enforce the given equation to be form invariant under an infinitesimal transformation involving independent and dependent variables. The infinitesimal transformations which leave the given equation form invariant are in turn called Lie point symmetries \cite{olv,Bluman}. Our primary interest here is to explore how these Lie point symmetries are intrinsically linked with the other integrability quantifiers.

\begin{definition}
The invariance of Eq.(\ref{main1}) under an one parameter group of Lie point symmetries, corresponding to the infinitesimal transformations 
\begin{equation}
T=t+\varepsilon \,\xi(t,x),~~~X=x+\varepsilon \,\eta(t,x),\quad \epsilon \ll 1,\label{asm1}
\end{equation}
where $\xi(t,x)$ and $\eta(t,x)$ are functions of their arguments and $\varepsilon$ is a small parameter, demands the following condition to be satisfied \cite{olv}:
\begin{equation}
\xi \frac{\partial \phi}{\partial t}+\eta \frac{\partial \phi}{\partial x}+\eta^{(1)}\frac{\partial \phi}{ \partial x^{(1)}}+...+\eta^{(n-1)}\frac{\partial \phi}{ \partial x^{(n-1)}}-\eta^{(n)}=0,\label{liec}
\end{equation} 
where $\eta^{(j)}$ is the $j^{th}$ prolongation, $j=1,2,3,...,n$, of the infinitesimal point transformations (\ref{asm1}) and defined to be  
\begin{equation}
\eta^{(1)}=\dot{\eta}-x^{(1)}\dot{\xi},~~\eta^{(2)}=\dot{\eta}^{(1)}-x^{(2)}\dot{\xi},...,~~\eta^{(n)}=\dot{\eta}^{(n-1)}-x^{(n)}\dot{\xi}.\label{prol}
\end{equation}
Here over dot denotes total differentiation with respect to $t$.
\end{definition}
Substituting the known expression $\phi$ in (\ref{liec}) and solving the resultant equation we can get the Lie point symmetries associated with the given $n^{th}$ order ODE. The associated vector field is given by $\hat{v}=\xi(t,x)\frac{\partial} {\partial t}+\eta(t,x)\frac{\partial} {\partial x}$. 

One may also introduce a characteristics 
\begin{equation}
Q=\eta-x^{(1)}\xi \label{charn}
\end{equation}
 and rewrite the invariance condition (\ref{liec}) in terms of a single variable $Q$ in the form
\begin{equation}  
D^n[Q]  = \phi_{x^{(n-1)}} D^{(n-1)}[Q]+...+\phi_{x^{(1)}} D[Q]+\phi_x Q.
\label{met16}
\end{equation}
One can deduce the coefficient functions $\xi$ and $\eta$ associated with the Lie point
symmetries from out of the class of solutions to (\ref{met16}) which depends only on $x$
and $t$, and also has a linear dependence in $x^{(1)}$. 

\begin{theorem} The connection between null forms in the extended Prelle-Singer procedure and Lie point symmetries is given by
\begin{equation}
D^{n-1}[Q]+S_{n-1}D^{n-2}D[Q]+\cdots+S_2D[Q]+S_1Q=0,\label{point_ps}
\end{equation}
where $D$ represents the total derivative operator.
\end{theorem}
\begin{proof}
Let $I(t,x,x^{(1)},\cdots,x^{(n-1)})$ be the first integral of the given $n^{th}$-order ODE (\ref{main1}). Then
\begin{equation}
\hat{v}^{(n-1)}I=\xi\frac{\partial I}{\partial t}+\eta \frac{\partial I}{\partial x}+\eta^{(1)}\frac{\partial I}{\partial x^{(1)}}+\cdots+\eta^{(n-1)}\frac{\partial I}{\partial x^{(n-1)}}=0\label{54po},
\end{equation}
where $\hat{v}^{(n-1)}=\xi\frac{\partial }{\partial t}+\eta \frac{\partial }{\partial x}+\eta^{(1)}\frac{\partial }{\partial x^{(1)}}+\cdots+\eta^{(n-1)}\frac{\partial }{\partial x^{(n-1)}}$. Rewriting Eq. (\ref{54po}) in terms of the characteristics function $Q(=\eta-x^{(1)}\xi)$ and using the Prelle-Singer procedure quantities $S_i$ and $R$ as $I_{t}  = R(\phi+\sum_{i=1}^{n-1}S_i x^{(i)}), ~I_{x}  = -RS_1,
~I_{x^{(i)}}  = -RS_{i+1},\;\;i=1,2,\ldots,n-2, ~I_{x^{(n-1)}} = -R$, we can obtain the expression as
\begin{equation}
(D^{n-1}[Q]+S_{n-1}D^{n-2}[Q]+\cdots+S_2D[Q]+S_1Q)+\xi(\frac{\partial I}{\partial t}+x^{(1)}\frac{\partial I}{\partial x}+\cdots+x^{(n)}\frac{\partial I}{\partial x^{(n-1)}})=0.
\end{equation}
In the above equation, the second term is the total derivative of the first integral and hence it vanishes. Thus we find
\begin{equation}
D^{n-1}[Q]+S_{n-1}D^{n-2}[Q]+\cdots+S_2D[Q]+S_1Q=0.\label{lieconn11}
\end{equation}
Equation (\ref{lieconn11}) establishes the connection between Lie point symmetries and null forms in the Prelle-Singer method. For $n=2$, the above relation reads $D[Q]+S_1Q=0$. Likewise the above relation for third-order ODEs turn out to be $D^2[Q]+S_2D[Q]+S_1Q=0$.
\end{proof}
\subsubsection{Other quantifiers from Lie point symmetries}
\label{sec_poin_other_meth}
Substituting the transformations (\ref{sn1eq}) and (\ref{ueq}) in (\ref{point_ps}) we can rewrite the latter equation in terms of $V$ and $X_i$ which upon solving yields these quantities. Substituting them back in Eqs.(\ref{sn1eq}) and (\ref{ueq}), we can obtain the null forms $S_i$ of the given equation. The associated integrating factor can be obtained by solving Eq.(\ref{net12}). As we demonstrated in our earlier works, the null forms and the integrating factor $R$ can be connected to adjoint-symmetries, Jacobi last multiplier, Darboux polynomials and $\lambda$-symmetries, and we can relate all these quantifiers recursively. Once we know the Lie point symmetries, $\lambda$-symmetries can be determined from (\ref{qlam}) which is given below. From the null forms and the integrating factor, we can find the generalized $\lambda$-symmetries using the relation (\ref{gen_vec_ps}) which is given in Sec. (\ref{ghtff}) below.

\subsection{Contact symmetries}
Several nonlinear ODEs do not admit Lie point symmetries but are proved to be integrable by other methods. To demonstrate the integrability of these nonlinear ODEs in the sense of Lie one should consider more generalized transformations. One such transformation includes velocity dependent (first derivative) terms in the infinitesimal transformations \cite{Bluman}. In the following, we give a brief account of the velocity dependent transformations and how they can be related with null forms and integrating factors that appear in the Prelle-Singer procedure. The last result is new to the literature.

Let us consider a one-parameter group of contact transformations \cite{Bluman}
\begin{equation}
T=t+\varepsilon \,\xi(t,x,x^{(1)}),~X=x+\varepsilon \,\eta(t,x,x^{(1)}),~\dot{X}=x^{(1)}+\varepsilon \, \eta^{(1)}(t,x,x^{(1)}) \quad \epsilon \ll 1,\label{asm}
\end{equation}
where $\xi(t,x,x^{(1)}),~\eta(t,x,x^{(1)})$ and $\eta^{(1)}(t,x,x^{(1)})$ are functions of their arguments and $\varepsilon$ is a small parameter. The functions $\xi$ and $\eta$ determine an infinitesimal contact transformation if it is possible to write them in the form \cite{ci7}
\begin{eqnarray}
\xi(t,x,\dot{x})=-\frac{\partial W} {\partial \dot{x}},~~\eta(t,x,\dot{x})=W-\dot{x}\frac{\partial W} {\partial \dot{x}},~~\eta^{(1)}=\frac{\partial W} {\partial t}+\dot{x}\frac{\partial W} {\partial x},
\end{eqnarray} 
where the characteristic function $W(t,x,\dot{x})$ is an arbitrary function of its arguments. If $W$ is linear in $\dot{x}$ the corresponding contact transformation is an extended point transformation and it holds that $W(t,x,\dot{x})=\eta(t,x)-\dot{x}\xi(t,x)$. The invariance of Eq.(\ref{main1}) under the infinitesimal contact transformation is given by 
\begin{equation}
\xi \frac{\partial \phi}{\partial t}+\eta \frac{\partial \phi}{\partial x}+\eta^{(1)}\frac{\partial \phi}{ \partial x^{(1)}}+...+\eta^{(n-1)}\frac{\partial \phi}{ \partial x^{(n-1)}}-\eta^{(n)}=0,\label{liecc}
\end{equation} 
where $\eta^{(j)}$ is the $j^{th}$ prolongation, $j=1,2,3,...,n$, of the infinitesimal transformation (\ref{asm}). The prolongations are defined as  
\begin{equation}
\eta^{(1)}=\dot{\eta}-x^{(1)}\dot{\xi},~~\eta^{(2)}=\dot{\eta}^{(1)}-x^{(2)}\dot{\xi},...,~~\eta^{(n)}=\dot{\eta}^{(n-1)}-x^{(n)}\dot{\xi},\label{prolc}
\end{equation}
where over dot denotes total differentiation with respect to $t$. The associated contact symmetry vector field is given by $\Omega=\xi(t,x,x^{(1)})\frac{\partial}{\partial t}+\eta(t,x,x^{(1)})\frac{\partial}{\partial x}$.

Analogous to the case of Lie point symmetries (see Eq.(\ref{charn})) one may introduce a characteristics \cite{Bluman}
\begin{equation}
W=\eta-x^{(1)}\xi \label{charnw}
\end{equation}
 and rewrite the invariance condition (\ref{liecc}) in terms of a single variable $W$ in the form
\begin{equation}  
D^n[W]  = \phi_{x^{(n-1)}} D^{(n-1)}[W]+...+\phi_{x^{(1)}} D[W]+\phi_x W.
\label{met16c}
\end{equation}
Solving Eq.(\ref{met16c}) one can get the characteristics $W$. From $W$ one can recover the contact symmetries $\xi$ and $\eta$. However, unlike the Lie point symmetries, it is very difficult to determine them systematically.

\begin{theorem} The connection between the null forms and the contact symmetries is given by
\begin{equation}
D^{n-1}[W]+S_{n-1}D^{n-2}D[W]+\cdots+S_2D[W]+S_1W=0,\label{contact_ps}
\end{equation}
where $D$ represents the total derivative operator.
\end{theorem}

\begin{proof}
The proof is analogous to the one which we have discussed in the earlier section, that is the connection between the Lie point symmetries and the Prelle-Singer procedure. Here we consider the contact symmetry characteristics $W$ instead of the Lie point symmetry characteristics $Q$.

\end{proof}

\subsubsection{Other quantifiers from contact symmetries}
Substituting Eqs.(\ref{sn1eq}) and (\ref{ueq}) in (\ref{contact_ps}) and solving it, we can find the expressions $V$ and $X_i$ from which one can derive the other quantifiers by following the steps given at the end of the previous sub-section. 

\subsection{$\lambda$-symmetries}
As noted above, the conventional Lie point symmetry analysis has been generalized in several directions such that the nonlinear ODEs which cannot solved by point symmetries can now be integrated with the help of suitable generalized transformations. Another such generalized symmetry is the $\lambda$-symmetry \cite{Mur3,Mur1,Mur4,gae1,gae2}. The components of these vector fields must satisfy a system of determining equations that depend on an arbitrary function $\lambda$, which can be chosen to solve the system easily. When this arbitrary function is chosen to be null, we obtain the classical Lie point symmetries. This method also provides a systematic procedure to find the first integrals and the integrating factors. 
\begin{definition}
Consider an $n^{th}$-order ODE which admits a $\lambda$-symmetry $\tilde{V}=\xi(t,x)\frac{\partial}{\partial t}+\eta(t,x)\frac{\partial}{\partial x}$ for some function $\lambda=\lambda(t,x,x^{(1)},...,x^{(n-1)})$, then the invariance of the $n^{th}$-order ODE under the  $\lambda$-symmetry vector field is given by \cite{Mur1}
\begin{equation}
\tilde{V}^{[\lambda,(n)]}(x^{(n)}-\phi(t,x,x^{(1)},...,x^{(n-1)}))|_{x^{(n)}=\phi}=0,\label{lam_invar}
\end{equation}
where $\tilde{V}^{[\lambda,(n)]}$ is given by $\xi \frac{\partial} {\partial t}+\eta \frac{\partial} {\partial x}+\eta^{[\lambda,(1)]}\frac{\partial} {\partial x^{(1)}}+...+\eta^{[\lambda,(n)]}\frac{\partial} {\partial x^{(n)}}$. Here $\eta^{[\lambda,(1)]}$, $\eta^{[\lambda,(2)]},...,\eta^{[\lambda,(n)]}$ are the first, second, $\cdots$, $n^{th}$ $\lambda$- prolongations respectively whose explicit expressions are given by \cite{Mur1}
\begin{eqnarray}
\eta^{[\lambda,(1)]}&=&(D+\lambda)\eta(t,x)-(D+\lambda)(\xi(t,x))x^{(1)},\label{etlam1} \nonumber\\
\eta^{[\lambda,(2)]}&=&(D+\lambda)\eta^{[\lambda,(1)]}(t,x,x^{(1)})-(D+\lambda)(\xi(t,x))x^{(2)},\label{etlam2}\nonumber \\
\vdots \nonumber \\
\eta^{[\lambda,(n)]}&=&(D+\lambda)\eta^{[\lambda,(n-1)]}(t,x,...,x^{(n-1)})-(D+\lambda)(\xi(t,x))x^{(n)}.\label{etlam3}
\end{eqnarray}
Expanding the $\lambda$-invariance condition (\ref{lam_invar}), we have
\begin{equation}
\xi\phi_t+\eta\phi_x+\eta^{[\lambda,(1)]}\phi_{x^{(1)}}+...+\eta^{[\lambda,(n-1)]}\phi_{x^{(n-1)}}-\eta^{[\lambda,(n)]}=0,
\label{beq1}
\end{equation}
where the $\lambda$-prolongations are given in (\ref{etlam3}).
\end{definition}
The $\lambda$-prolongation (\ref{etlam3}) reduces to the classical Lie point prolongation formula (\ref{liec}) when $\lambda=0$. Solving the invariance condition (\ref{beq1}) we can obtain the explicit forms of $\xi, \eta$ and $\lambda$.


Suppose the given ODE admits Lie point symmetries, then the $\lambda$-symmetries can be derived without solving the invariance condition (\ref{beq1}). The $\lambda$-symmetries can be directly obtained from the Lie point symmetries through the expression \cite{Mur3}
\begin{equation}
\lambda=\frac{D[Q]} {Q},\label{qlam}
\end{equation}
where $D$ is the total differential operator and $Q=\eta-x^{(1)}\xi$ provided the expression given in (\ref{qlam}) satisfies (\ref{beq1}). The associated $\lambda$-symmetry vector field is given by $\tilde{V}=\frac{\partial} {\partial x}$. One can also construct more number of $\lambda$-symmetry vector fields associated with the $n^{th}$-order ODE (\ref{main1}) by solving the invariance condition (\ref{beq1}).

\subsubsection{Connection between $\lambda$-symmetries and Prelle-Singer method}
The interconnection between $\lambda$-symmetries and the Prelle-Singer method for second- and third-order ODEs were discussed elaborately in Refs.\cite{suba1,suba2,Mur3}. Here we generalizes the interlink to $n^{th}$-order ODEs. 
\begin{theorem}
The null forms associated with the $n^{th}$-order ODE (\ref{main1}) in the extended Prelle-Singer procedure can be connected to the $\lambda$-symmetries through the relation
\begin{equation}
 S_1+\sum_{i=1}^{n-1}((D+\lambda)^{i-1}\lambda) S_{i+1}=0,\label{lamus}
\end{equation}
with the assumption that $S_n=1$.
\end{theorem}

\begin{proof}
Let $I(t,x,...,x^{(n-1)})$ be the first integral of (\ref{main1}) then $R=-I_{x^{(n-1)}}$ is an integrating factor. The total derivative $\frac{dI} {dt}=0$ gives
\begin{equation}
R \phi=I_t+x^{(1)} I_x+...+x^{(n-1)}I_{x^{(n-2)}}.
\end{equation}
Let this $I(t,x,x^{(1)},...,x^{(n-1)})$ also be a first integral of $\tilde{V}^{[\lambda,(n-1)]}$ for some function $\lambda(t,x,x^{(1)},...,x^{(n-1)})$, then 
\begin{equation}
\tilde{V}^{[\lambda,(n-1)]}I=I_x+\eta^{(1)}I_{x^{(1)}}+...+\eta^{(n-1)}I_{x^{(n-1)}}=0,\label{exlam}
\end{equation}
when $\lambda$ is such that $\tilde{V}=\frac{\partial}{\partial x}$ is a $\lambda$-symmetry. Substituting the expressions
\begin{equation}
\eta^{(i)}=(D+\lambda)^{j}\lambda,~~i=j+1,~~j=0,1,2,...,n-2,\label{lampro}
\end{equation}
in Eq.(\ref{exlam}), we get
\begin{equation}
I_x+\lambda I_{x^{(1)}}+(D+\lambda)\lambda I_{x^{(2)}}+...+(D+\lambda)^{n-2}\lambda I_{x^{(n-1)}}=0.\label{ejg}
\end{equation}
Recalling the relations $I_{x}=-RS_1,I_{x^{(i)}}=-RS_{i+1},i=1,2,\ldots,n-2,
I_{x^{(n-1)}} = -R$ from the Prelle-Singer method and substituting them in (\ref{ejg}) and rewriting the later we will end up with (\ref{lamus}).
\end{proof} 
The above result shows that in the case of second-order ODEs, since we have only one null form $S_1$ it is directly connected with the $\lambda$-symmetry through the expression $\lambda=-S_1$ \cite{suba1,Mur3}. In the case of  third-order ODEs, we have two null forms, namely $S_1$ and $S_2$, which can be connected to the $\lambda$-symmetries through the differential relation $D[\lambda]+\lambda^2+S_2\lambda+S_1=0$ \cite{suba2}. From (\ref{lamus}) we infer that in fourth-order ODEs, $\lambda$-symmetries and null forms are connected by the differential relation $(D^2[\lambda]+3\lambda D[\lambda]+\lambda^3)+(D[\lambda]+\lambda^2)S_3+\lambda S_2+S_1=0$. The interconnection persists in all orders and for an $n^{th}$-order ODE the explicit expression is given by (\ref{lamus}). We note here that while deriving (\ref{lamus}) we assumed that $\lambda \neq 0$. In the case $\lambda =0$ we have $I_x=0$ (vide Eq.(\ref{ejg})). 
    
\subsubsection{Other quantifiers from $\lambda$-symmetries}
\label{lam_inet}

From the known $\lambda$-symmetries and by substituting the transformations (\ref{sn1eq}) and (\ref{ueq}) in the expression (\ref{lamus}), we can rewrite the latter equation in terms of $V$ and $X_i$ which upon solving yields these quantities. Substituting them back in Eqs.(\ref{sn1eq}) and (\ref{ueq}), we can obtain the null forms $S_i$ of the given equation. From the null forms, we can obtain the integrating factor through (\ref{net12}). From the Prelle-Singer method quantities we can identify the rest of the integrability quantifiers. 


\subsection{Jacobi last multiplier method}
The Jacobi last multiplier method is yet another important analytical method to prove the integrability of the given dynamical system \cite{jac,jac1}. This method helps to determine the integrals associated with the given equation. The multipliers can also be used to find the Lagrangians of the associated ODE whenever the considered ODE is of even order \cite{Nuc}.  
\begin{definition}
Let us assume that the ODE (\ref{main1}) admits a first integral \\$I(t,x,x^{(1)},...,x^{(n-1)})=M_1/M_2=C$, where $C$ is a constant on solutions, and $M_i,~i=1,2,$'s are the associated Jacobi last multipliers. These multipliers can be determined for the given $n^{th}$-order ODE by solving the following determining equation \cite{Nuc}
\begin{equation}  
D[\log M]+\phi_{x^{(n-1)}}=0.\label{met20}
\end{equation}
Here $D$ represents the total differential operator.~
Substituting the given equation in (\ref{met20}) and solving the resultant equation we can obtain the last multipliers associated with a $n^{th}$ order ODE. The ratio of any two Jacobi last multipliers, say $M_1/M_2$, gives the first integral. Note that $M_1$ and $M_2$ may also be trivially related, that is $M_1=\alpha M_2$, where $\alpha$ is a constant number.
\end{definition}

In the following sub-sections, the connection between Jacobi last multiplier, Lie point symmetries and the Prelle-Singer method is discussed. 

\subsubsection{Jacobi last multiplier (JLM) and Lie point symmetries}
The connection between Lie point symmetries and JLM is known for a long time \cite{Nuc}. It is given by
\begin{equation}
M=\frac{1} {\Delta},~~\Delta \neq 0,\label{hnn}
\end{equation}
where
\begin{equation}
\Delta = \begin{vmatrix}
1 & x^{(1)} & \cdots & x^{(n)} \\
\xi_1 & \eta_1 & \cdots & \eta_{1}^{(n-1)} \\
\vdots & \vdots & \vdots & \vdots \\
\xi_n & \eta_n & \cdots & \eta_{n}^{(n-1)}\end{vmatrix},\label{delta}
\end{equation}
where $(\xi_1,\eta_1)$, $(\xi_2,\eta_2),\cdots,(\xi_n,\eta_n)$  are the Lie point symmetries of the $n^{th}$-order ODE, $\eta_{1}^{(i)}$, $\eta_{2}^{(i)},\cdots,\eta_{n}^{(i)}$, $i=1,2,\cdots, n-1$, are their corresponding prolongations, respectively, and the inverse of $\Delta$ defines the multiplier of the given equation provided that $\Delta \neq 0$. Note that there are cases where Jacobi last multipliers exist even without Lie point symmetries. In such cases the existence of (\ref{hnn}) may not be valid. However, the other quantifiers can be determined as pointed out below in Sec. (\ref{jlm_inten}).

\subsubsection{Jacobi last multiplier and Prelle-Singer method}
The connection between the Jacobi last multiplier and the extended Prelle-Singer method is shown below.
\begin{theorem} The connection between the integrating factor $R$ of the extended Prelle-Singer procedure and the Jacobi last multiplier $M$ is given by
\begin{equation}
R=VM,\label{jlm_ps}
\end{equation}
where $V$ is the function defined by (\ref{s_tran_eqn}).
\end{theorem}
\begin{proof}
By comparing Eq.(\ref{met20}) with (\ref{met15}), we observe that
\begin{equation}
F=\frac{1} {M}.\label{sssum1}
\end{equation}
Substituting (\ref{sssum1}) into Eq.(\ref{rvf}), we obtain
\begin{equation}
R=VM.
\end{equation}

Since $R$ and $V$ are known from the Prelle-Singer method, the Jacobi last multiplier can be derived from the Prelle-Singer method itself. This establishes the connection between the Prelle-Singer method and the Jacobi last multiplier method.
\end{proof}
\subsubsection{Other quantifiers from Jacobi last multipliers}\label{jlm_inten}
Suppose if we know the Jacobi last multipliers, then the ratio of the multipliers gives the first integral. Once we know the integral, we can identify the integrating factor $R$ through the relation $R=-I_{x^{(n-1)}}$. Once we know the integrating factor, we can obtain the null forms $S_i$ from Eqs.(\ref{net9}), (\ref{net10}) and (\ref{net11}). We can determine the $\lambda$-symmetries from the null forms through the relation (\ref{lamus}). The rest of the quantities can be constructed as outlined in Sec. (\ref{lam_inet}).

\subsection{Darboux polynomials approach}
Darboux theory of integrability is yet another approach which helps to determine the integrals of the given ODE. The ratio of two Darboux polynomials gives an integral provided they have the same cofactor \cite{darb}. The Darboux theory is extensively studied in the contemporary literature \cite{dumor}. 

\begin{definition}
Let us assume that the ODE (\ref{main1}) admits a first integral \\$I(t,x,x^{(1)},...,x^{(n-1)})=C$, where $C$ is a constant on solutions. Darboux polynomial determining equation for an $n^{th}$-order ODE (\ref{main1}) is given by \cite{darb}
\begin{equation}
D[f]=g(t,x,...,x^{(n-1)})f,\label{Darb}
\end{equation}
where $D$ is the total differential operator and $g(t,x,...,x^{(n-1)})$ is the cofactor. Solving Eq.(\ref{Darb}) through appropriate ansatz on $f$ and $g$, we can obtain the Darboux polynomials $(f)$ and the cofactor $g$. The ratio of two Darboux polynomials $f_1/f_2$ which shares the same cofactor $g$ gives a first integral.
\end{definition}
 The combinations of the Darboux polynomials also define a first integral. For example, let us consider the function $G=\prod_i f_i^{n_i}$, where $f_i'$s are the Darboux polynomials and $n_i'$s are rational numbers. If we can identify a sufficient number of Darboux polynomials,  $f_i'$s, satisfying the relations $D[f_i]/f_i=g_i$, where $g_i$'s are the co-factors, then 
\begin{equation}
D[G]/G=\sum_in_i\frac{D[f_i]}{f_i}=\sum_in_ig_i.\label{laseq}
\end{equation}
We use this property to derive the connection between the Darboux polynomials and the Prelle-Singer procedure.

\begin{cor} The connection between the integrating factor $R$ of the extended Prelle-Singer procedure and the Darboux polynomials is given by
\begin{equation}
R=\frac{V} {F}\label{jlm_df}
\end{equation}
 provided the cofactor of the Darboux polynomial is $\phi_{x^{(n-1)}}$.
\end{cor}             
\begin{proof}
Fixing $G=F$ in Eq.(\ref{laseq}) and comparing the later equation with Eq.(\ref{met15}), we find the cofactor of the Darboux polynomial equation as
\begin{equation}
\sum_in_i\frac{D[f_i]}{f_i}=\phi_{x^{(n-1)}}.
\end{equation}
Substituting the Darboux polynomial $G=F$ with the cofactor $\phi_{x^{(n-1)}}$ in Eq.(\ref{rvf}), we obtain $R=V/F$. 

\end{proof}

\subsubsection{Connection between Darboux polynomials and Jacobi last multiplier}
\label{dpandjlm}
By comparing Eqs.(\ref{Darb}) and (\ref{met20}) we can identify the connection between the Darboux polynomials and the Jacobi last multiplier, provided its cofactor $g=\phi_{x^{(n-1)}}$, as
\begin{equation}
M=F^{-1}.\label{Drbre}
\end{equation}
Using the above relation, we can deduce the Jacobi last multiplier from the Darboux polynomials themselves. We note here that this relation is already known in the literature for $n^{th}$-order ODEs provided the cofactor of Darboux polynomials is $\phi_{x^{(n-1)}}$ \cite{suba1}.

\subsubsection{Other quantifiers from Darboux polynomials} 
Darboux polynomials are given as starters, then the ratio of the Darboux polynomials gives the first integral. Once we know the integral, we can identify the integrating factor $R$ through the relation $R=-I_{x^{(n-1)}}$. From the known integrating factor, we can obtain the null forms $S_i$ from Eqs.(\ref{net9}), (\ref{net10}) and (\ref{net11}). From the Prelle-Singer method quantities, we can find the other integrability quantifiers.

 
\subsection{Adjoint-symmetries}
An integrating factor is a set of functions, which on multiplying the given ODE yields a first integral. If the system is self-adjoint, then its integrating factors are necessarily solutions of its linearized system (\ref{met16}). Such solutions are the symmetries of the given system. If the given ODE is not self-adjoint, then its integrating factors are necessarily solutions of the adjoint system of its linearized system. Such solutions are known as adjoint-symmetries of the given ODE \cite{Bluman,blu_pap}.
\begin{definition}
The adjoint ODE of the linearized symmetry condition (\ref{met16}) can be written as 
\begin{equation}
D^n[\Lambda]+D^{n-1}[\phi_{x^{(n-1)}}\Lambda]-D^{n-2}[\phi_{x^{(n-2)}}\Lambda]+...+(-1)^{n-1}\phi_{x}\Lambda=0,\label{met23}
\end{equation}
where $D$ and $\Lambda$ represent the total derivative operator and adjoint-symmetries, respectively. Solutions of Eq.(\ref{met23}) are called adjoint symmetries \cite{blu_pap}.
\end{definition}
 If these solutions also satisfy the adjoint-invariance condition
\begin{equation}
\Lambda_{x^{(n-2m)}}+\sum_{i=1}^{2m-1}(-1)^{i-1}(D^{i-1}(\phi_{x^{(n-2m+i)}}\Lambda))_{x^{(n-1)}}+(D^{2m-1}\Lambda)_{x^{(n-1)}}=0,m=1,\cdots,\frac{n}{2}
\label{adjo1}
\end{equation}
then they become integrating factors \cite{blu_pap}.
 Once we know the integrating factors, we can construct the first integrals. Multiplying the given equation by these integrating factors and rewriting the resultant equation as a perfect differentiable function, we can identify the first integrals, that is
\begin{equation}
\Lambda_i(t,x,x^{(1)},...,x^{(n-1)})(x^{(n)}-\phi(t,x,x^{(1)},...,x^{(n-1)}))=\frac{d} {dt}\big(I_i\big),~i=1,2,...,m.\label{adjint}
\end{equation}


\subsubsection{Adjoint-symmetries and Prelle-Singer method}
The connection between adjoint-symmetries and Prelle-Singer method is given below:
\begin{theorem}
The connection between adjoint symmetries and the Prelle-Singer method is given by
\begin{equation}
R=\Lambda.\label{r_lam_con}
\end{equation}
\end{theorem}
\begin{proof}
Rewriting the Eqs.(\ref{net9}),~~(\ref{net10}),~~(\ref{net11}) and (\ref{net12}) as a single equation in one variable, $R$, we find
\begin{equation}  
D^n[R]+D^{n-1}[\phi_{x^{(n-1)}}R]-D^{n-2}[\phi_{x^{(n-2)}}R]+...+(-1)^{n-1}\phi_{x}R=0.\label{met24}
\end{equation}
Comparing the adjoint of the linearized symmetry condition equation (\ref{met23}) and (\ref{met24}) one can conclude that the integrating factor $R$ is nothing but the adjoint symmetry $\Lambda$, that is 
\begin{equation}
R=\Lambda. \label{met25}
\end{equation}
Thus the integrating factor turns out to be the adjoint symmetry of the given $n^{th}$ order nonlinear ODE.
\end{proof}
\subsubsection{Other quantifiers from adjoint-symmetries}
Since adjoint-symmetries are also integrating factors provided they satisfy Eq. (\ref{adjo1}), we can find the null forms from them through the relations (\ref{net9}), (\ref{net10}) and (\ref{net11}). From the null forms, we can find the $\lambda$-symmetries using Eq.(\ref{lamus}). From the $\lambda$-symmetries, the expressions $V$ and $X_i$ can be obtained from the null forms by utilizing the expressions (\ref{sn1eq}) and (\ref{ueq}). We can also obtain the Lie point symmetries with the help of (\ref{lieconn11}). From a knowledge of the integrating factor $R$ and the function $V$, the Darboux polynomials can be determined from (\ref{rvf}). Then the Jacobi last multiplier can also be found from the inverse of the Darboux polynomials. 

\subsection{Generalized $\lambda$-symmetries}
\label{ghtff}
Generalized $\lambda$-symmetries contain the well-known subclasses of vector
fields that have appeared in the literature. The coefficient functions in the infinitesimal generators are not only functions of the dependent and independent variables but also functions of their derivatives upto the order $n-1$ and an arbitrary function $\lambda$. For more details on these generalized $\lambda$-symmetries one may refer \cite{Mur4}. 

\begin{definition}
The generalized vector field is given by
\begin{eqnarray}
 \gamma^{(n)}&=&\xi(t,x,x^{(1)},\cdots,x^{(n-1)})\frac{\partial}{\partial t}+\eta(t,x,x^{(1)},\cdots,x^{(n-1)})\frac{\partial}{\partial x}\nonumber \\ &&+\sum_{i=1}^{n}\zeta^{(i)}(t,x,x^{(1)},\cdots,x^{(n-1)})\frac{\partial}{\partial x^{(i)}}.
\end{eqnarray}
Here $\xi(t,x,x^{(1)},\cdots,x^{(n-1)})$, $\eta(t,x,x^{(1)},\cdots,x^{(n-1)})$ and $\zeta^{(1)}(t,x,x^{(1)},\cdots,x^{(n-1)})$ are the three arbitrary functions which we have to determine. Their prolongations are given by the following expressions \cite{Mur4}
\begin{equation}
 \zeta^{(i)}=D[\zeta^{(i-1)}]-x^{(i)}D[\xi]+\frac{\zeta^{(1)}+x^{(1)}D[\xi]-D[\eta]}{\eta-x^{(1)}\xi}(\zeta^{(i-1)}-x^{(i)}\xi),\label{zet2}
\end{equation}
where $i$ takes the value from $2$ to $n$. Suppose $\xi$, $\eta$ and $\zeta^{(1)}$ are functions of $t$, $x$ and $x^{(1)}$ alone, then the generalized $\lambda$-symmetries are called telescopic vector fields \cite{pucci}.  We note here that the function $\lambda$ is given by $\lambda=(\zeta^{(1)}+x^{(1)}D[\xi]-D[\eta])/(\eta-x^{(1)}\xi)$. 
\end{definition}

\begin{theorem}
The generalized $\lambda$-symmetries are related with the integrating factors and null forms of the Prelle-Singer method through the expression
\begin{equation}
 RS_1+\sum_{i=1}^{n-1}\zeta^{(i)}S_{i+1}=0,\label{gen_vec_ps}
\end{equation}
where for $i+1=n$, we assume that $S_n=1$. The underlying generalized $\lambda$-symmetry vector field reads
\begin{equation}
 R\frac{\partial}{\partial x}+\sum_{i=1}^{n}\zeta_i\frac{\partial}{\partial x^{(i)}}.
\end{equation}
\end{theorem}

\begin{proof}
 Let $I(t,x,x^{(1)},...,x^{(n-1)})$ be the first integral of the given $n^{th}$-order ODE, then
 \begin{equation}
  \gamma^{(n-1)}I=0.
 \end{equation}
Expanding the above expression, we get 
\begin{equation}
 \xi I_t+\eta I_x+\zeta^{(1)}I_{x^{(1)}}+\cdots+\zeta^{(n-1)}I_{x^{(n-1)}}=0.\label{25fd}
\end{equation}
We have three unknowns to obtain the generalized $\lambda$-symmetry vector field. It is very difficult to integrate the above Eq.(\ref{25fd}). As our aim is to interrelate the generalized $\lambda$-symmetries with the Prelle-Singer procedure we consider $\xi=0$ and $\eta=R$ and substitute the expressions $I_{x}=-RS_1,I_{x^{(i)}}=-RS_{i+1},i=1,2,\ldots,n-2,
I_{x^{(n-1)}} = -R$ (which appear in the Prelle-Singer procedure) in (\ref{25fd}) and rewrite the latter expression suitably to obtain equation (\ref{gen_vec_ps}). The generalized $\lambda$-symmetries turn out to be of the form
\begin{equation}
 \gamma^{(n)}=R\frac{\partial}{\partial x}+\zeta^{(1)}\frac{\partial}{\partial x^{(1)}}+\cdots+\zeta^{(n)}\frac{\partial}{\partial x^{(n)}}.
\end{equation}
\end{proof}


From the generalized $\lambda$-symmetries, we can find the integrating factors and null forms using the relation (\ref{gen_vec_ps}). Once we know the null forms and integrating factors we can determine all other quantifiers as we described earlier.

\begin{figure}[h!]
  \centering
   \includegraphics[width=0.75\textwidth]{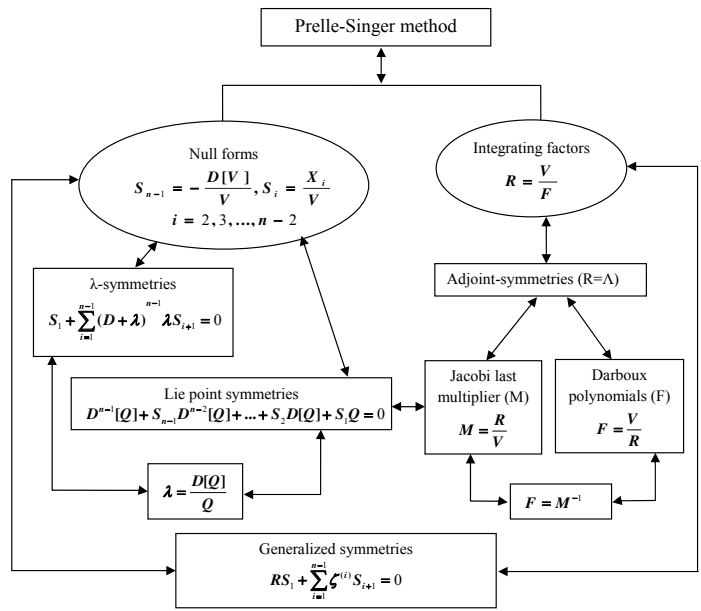}
    \caption{Flow chart connecting Prelle-Singer procedure with other methods for $n^{th}$-order ODEs}
\end{figure}
The broader interconnections which exist between the eight different analytical methods in the literature is summarized in Figure. 1. As we pointed out earlier in the extended Prelle-Singer procedure we have two sets of quantities, namely the null forms and integrating factors, which determine the integrability. How these two quantities are interrelated with other quantifiers is schematically given in Figure 1. The null forms $S_i$ can be  connected through Lie point symmetries, contact symmetries (where $Q$ is replaced by $W$) and $\lambda$-symmetries. On the other hand, the adjoint-symmetries, Jacobi last multiplier and Darboux polynomials can be brought out from the integrating factor. Generalized $\lambda$-symmetries can be related with the Prelle-Singer procedure through the null forms and integrating factor. We wish to note here that the cases where no Lie point symmetries exist, that part of the Figure 1 alone will be absent. Other interconnections between the integrability quantifiers will remains same as in Figure 1.

\section{Examples}
In this section, we prove the above said interconnections by considering an example each in second-order, third-order, fourth-order and $n^{th}$-order differential equations. As we generalize our earlier results, we do not discuss the method of finding all the quantities for second- and third-order ODEs. Instead of that for these two cases we report the method of finding contact symmetries and generalized $\lambda$-symmetries from the Prelle-Singer procedure which we have not discussed earlier. The interconnections for the higher-order ODEs ($n>3$) are demonstrated explicitly.

\subsection{Second-order ODE}
We start our illustration with the same second-order nonlinear ODE which we considered in our earlier paper, namely \cite{suba1}
\begin{equation}
x^{(2)}+3xx^{(1)}+x^3=0.\label{second_example}
\end{equation}
In our earlier work \cite{suba1}, we have studied the interconnection between Prelle-Singer method with Lie point symmetries, Jacobi last multiplier, Darboux polynomial, adjoint-symmetries and $\lambda$-symmetries. In this work, we explore the contact symmetries and the telescopic vector fields of (\ref{second_example}) from the null forms and integrating factors. We recall them from Ref.\cite{suba1}:
\begin{eqnarray}
S_{11}&=&x-\frac{x^{(1)}} {x},\;\;\;\;\;\;\;S_{12}=\frac{(2-x^{(1)}t^2-4tx+t^2x^2)} {t(-2+tx)},\label{s12}\\
R_1&=&\frac{x} {(x^2+x^{(1)})^2},\;\;\;R_2=\frac{-t(1-\frac{tx} {2})} {(tx^{(1)}-x+tx^2)^2}.\label{r22}
\end{eqnarray}
Substituting the null forms (\ref{s12}) in (\ref{contact_ps}) we get
\begin{eqnarray}
 D[W_1]+\bigg(x-\frac{x^{(1)}} {x}\bigg) W_1=0~~\mathrm{and}~~
 D[W_2]+\bigg(\frac{2-x^{(1)}t^2-4tx+t^2x^2} {t(-2+tx)}\bigg)W_2=0. 
\end{eqnarray}
Solving the above equations, we can get the following characteristics, namely
\begin{equation}
 W_1=\frac{x \left(x^2+2 x^{(1)}\right)}{x^2+x^{(1)}},~~W_2=\frac{t (t x-2) \left(x^2+2x^{(1)}\right)}{x^2+x^{(1)}}.
\end{equation}
Using the identity $W=\eta-x^{(1)}\xi$, we can determine the contact symmetries of (\ref{second_example}). The associated vector field reads 
\begin{equation}
 \Omega_1=\frac{x \left(x^2+2 x^{(1)}\right)}{x^2+x^{(1)}}\frac{\partial}{\partial x},~\Omega_2=\frac{t (t x-2) \left(x^2+2x^{(1)}\right)}{x^2+x^{(1)}}\frac{\partial}{\partial x}.
\end{equation}
Substituting the null forms (\ref{s12}) and the integrating factors (\ref{r22}) in (\ref{gen_vec_ps}), we obtain ($\zeta^{(1)}=-RS$)
\begin{equation}
 \zeta_1^{(1)}=\bigg(\frac{x^2-x^{(1)}}{\left(x^2+x^{(1)}\right)^2}\bigg),~\zeta_2^{(1)}=\bigg(\frac{t^2 \left(x^{(1)}-x^2\right)+4 t x-2}{\left(t^2 \left(x^2+x^{(1)}\right)-2 t x+2\right)^2}\bigg).
\end{equation}
The corresponding generalized $\lambda$-symmetries can be written as (telescopic vector fields)
\begin{eqnarray}
\hspace{-0.5cm}\gamma_1^{(2)}&=&-\bigg(\frac{x}{\left(x^2+x^{(1)}\right)^2}\bigg)\frac{\partial}{\partial x}+\bigg(\frac{x^2-x^{(1)}}{\left(x^2+x^{(1)}\right)^2}\bigg)\frac{\partial}{\partial x^{(1)}}+\bigg(\frac{6 x x^{(1)}}{\left(x^2+x^{(1)}\right)^2}\bigg)\frac{\partial}{\partial x^{(2)}},\label{tel_com11}\nonumber\\
\hspace{-1cm}\gamma_2^{(2)}&=&-\bigg(\frac{t (2-t x)}{\left(t^2 \left(x^2+x^{(1)}\right)-2 t x+2\right)^2}\bigg)\frac{\partial}{\partial x}+\bigg(\frac{t^2 \left(x^{(1)}-x^2\right)+4 t x-2}{\left(t^2 \left(x^2+x^{(1)}\right)-2 t x+2\right)^2}\bigg)\frac{\partial}{\partial x^{(1)}}\nonumber \\
\hspace{-1cm}&&-\bigg(\frac{6 (t x-1) (tx^{(1)}+x)}{\left(t^2 \left(x^2+x^{(1)}\right)-2 t x+2\right)^2}\bigg)\frac{\partial}{\partial x^{(2)}}.\label{tel_com21}
\end{eqnarray}
We note here that obtaining the generalized $\lambda$-symmetries by solving the invariance condition associated with it is not a simple task. We have succeeded in determining them by establishing an interconnection with the other quantifiers.

\subsection{Third-order ODEs}
To demonstrate the interconnection in third-order ODEs, we again consider the same example which we considered in our earlier work \cite{suba2} that is  
\begin{equation}
\dddot{x}=\frac{6t\ddot{x}^3} {(x^{(1)})^2}+\frac{6\ddot{x}^2} {x^{(1)}}.\label{exam1}
\end{equation}
As in the previous example we intend to obtain the contact symmetries and generalized $\lambda$-symmetries from other quantifiers which are already known for this equation. We recall the null forms and the integrating factors from our previous work (Ref. \cite{suba2}) which are given by
\begin{eqnarray}
 &&S_{11}=\frac{2 (x^{(2)})^2}{(x^{(1)})^2},~~S_{12}=0,~~S_{13}=0,\label{third_eqn}\\
 &&S_{21}=-\frac{2 x^{(2)} (3 t x^{(2)}+x^{(1)})}{(x^{(1)})^2},~~S_{22}=-\frac{3 x^{(2)} (2 t x^{(2)}+x^{(1)})}{(x^{(1)})^2}\nonumber\\
 &&S_{23}=-\frac{ 6 t (x^{(2)})^2+3 x^{(1)}x^{(2)}}{(x^{(1)})^2},\\
 &&R_1=\frac{(x^{(1)})^2}{(x^{(2)})^2},~~R_2=\frac{(x^{(1)})^3}{(x^{(2)})^2},~~R_3=\frac{(x^{(1)})^4}{(x^{(2)})^2}.
\end{eqnarray}
Substituting the above null forms in (\ref{contact_ps}) we obtain
\begin{eqnarray}
D^2[W_1]-\frac{2 x^{(2)} (3 t x^{(2)}+x^{(1)})}{(x^{(1)})^2}D[W_1]+\frac{2 (x^{(2)})^2}{(x^{(1)})^2}W_1&=&0,\\
D^2[W_2]-\frac{3 x^{(2)} (2 t x^{(2)}+x^{(1)})}{(x^{(1)})^2}D[W_1]&=&0,\\
D^2[W_3]-\frac{ 6 t (x^{(2)})^2+3 x^{(1)}x^{(2)}}{(x^{(1)})^2}D[W_3]&=&0.
\end{eqnarray}
Solving the above equations, we obtain the characteristics as
\begin{equation}
 W_1=\frac{1}{x^{(1)}},~~W_2=\frac{1}{(x^{(1)})^2},~~W_3=t (x^{(1)})^2.
\end{equation}
By rewriting the above characteristics in the form $W=\eta-\dot{x}\xi$, we get the contact symmetries as ($\xi=0$ for first two symmetries, $\eta=0$ for third symmetry)
\begin{equation}
 \Omega_1=\frac{1}{x^{(1)}}\frac{\partial}{\partial x},~~\Omega_2=\frac{1}{(x^{(1)})^2}\frac{\partial}{\partial x},~~\Omega_3=-tx^{(1)}\frac{\partial}{\partial t}.
\end{equation}
The generalized $\lambda$-symmetry vector field component $\zeta^{(1)}$ can be obtained by substituting the null forms and the integrating factors in Eq. (\ref{gen_vec_ps}). The latter equation turns out to be
\begin{eqnarray}
 D[\zeta_1^{(1)}]+\frac{(x^{(2)})^2}{(x^{(1)})^2} (\zeta^{(1)})^2+\bigg(\frac{2 x^{(2)} (3 t x^{(2)}+4 x^{(1)})}{(x^{(1)})^2}\bigg)\zeta^{(1)} +2&=&0,\\
 D[\zeta_2^{(1)}] + \frac{(x^{(2)})^2}{(x^{(1)})^3}(\zeta^{(1)})^2 + \bigg(\frac{6 x^{(2)} (t x^{(2)}+x^{(1)})}{(x^{(1)})^2}\bigg)\zeta^{(1)}&=&0,\\
 D[\zeta_3^{(1)}] +\frac{(x^{(2)})^2}{(x^{(1)})^4}(\zeta^{(1)})^2 +\bigg(\frac{4 x^{(2)} (3 t x^{(2)}+2 x^{(1)})}{(x^{(1)})^2}\bigg)\zeta^{(1)}&=&0. 
\end{eqnarray}
Solving the above equations, we find
\begin{equation}
 \zeta_1^{(1)}=\frac{x^{(1)} (t x^{(2)}+x^{(1)})}{t (x^{(2)})^2},~~\zeta_2^{(1)}=-\frac{2 (x^{(1)})^2}{x^{(2)}},~~
 \zeta_3^{(1)}=-\frac{(x^{(1)})^3}{x^{(2)}}.
\end{equation}
From the above functions and the integrating factor $R$ and using the expressions (\ref{zet2}), we can get the other components $\zeta^{(2)}$ and $\zeta^{(3)}$ in the generalized $\lambda$-symmetries. 
The generalized $\lambda$-symmetries are given by
\begin{eqnarray}
 \gamma_1^{(3)}&=&\frac{(x^{(1)})^2}{(x^{(2)})^2}\frac{\partial}{\partial x}+(\frac{x^{(1)} (t x^{(2)}+x^{(1)})}{t (x^{(2)})^2})\frac{\partial}{\partial x^{(1)}}+(\frac{2 x^{(1)}}{t x^{(2)}}+\frac{6 t x^{(2)}}{x^{(1)}}+6)\frac{\partial}{\partial x^{(2)}}\nonumber \\
 &&+6 \left(\frac{18 t^2 (x^{(2)})^3}{(x^{(1)})^3}+\frac{28 t (x^{(2)})^2}{(x^{(1)})^2}+\frac{3}{t}+\frac{15 x^{(2)}}{x^{(1)}}\right)\frac{\partial}{\partial x^{(3)}},\nonumber \\
 \gamma_2^{(3)}&=& \frac{(x^{(1)})^3}{(x^{(2)})^2}\frac{\partial}{\partial x}-\frac{2 (x^{(1)})^2}{x^{(2)}}\frac{\partial}{\partial x^{(1)}}-6 (2 t x^{(2)}+x^{(1)})\frac{\partial}{\partial x^{(2)}} \nonumber \\&&
 -\frac{12 x^{(2)} \left(18 t^2 (x^{(2)})^2+19 t x^{(1)} x^{(2)}+5 (x^{(1)})^2\right)}{(x^{(1)})^2} \frac{\partial}{\partial x^{(3)}},\nonumber \\
 \gamma_3^{(3)}&=& \frac{(x^{(1)})^4}{(x^{(2)})^2}\frac{\partial}{\partial x} -\frac{(x^{(1)})^3}{x^{(2)}} \frac{\partial}{\partial x^{(1)}}-2 x^{(1)} (3 t x^{(2)}+2 x^{(1)}) \frac{\partial}{\partial x^{(2)}} \nonumber \\ &&-\frac{6 x^{(2)} \left(18 t^2 (x^{(2)})^2+22 t x^{(1)} x^{(2)}+7 (x^{(1)})^2\right)}{x^{(1)}} \frac{\partial}{\partial x^{(3)}}.
\end{eqnarray}
We have obtained the generalized $\lambda$-symmetries without solving the invariance condition with the help of established interconnections. 
\subsection{Fourth-order ODEs}
The interconnections in fourth-order ODEs are  yet to be demonstrated in the literature. So we intend to consider an example which is governed by the fourth-order ODE, namely \cite{Bluman}
\begin{equation}
x^{(4)}=\frac{4}{3}\frac{(x^{(3)})^2}{x^{(2)}}.\label{fourth_example}
\end{equation}
Equation (\ref{fourth_example}) was discussed in Ref. \cite{Bluman}, where the authors have studied second-order symmetries, Lie point symmetries and integrals of the above equation. The other quantifiers are not known for this equation. 

To obtain the null forms in the Prelle-Singer procedure, we have to solve the determining equations  (\ref{net9})-(\ref{net11}). Doing so, we find
\begin{eqnarray}
\small
&&\hspace{-1.2cm}S_{11}=0,~~S_{21}=\frac{(x^{(3)})^2}{6(x^{(2)})^2-3x^{(1)}x^{(3)}},~~S_{31}=\frac{x^{(3)}(-9(x^{(2)})^2+4x^{(1)}x^{(3)})}{6(x^{(2)})^3-3x^{(1)}x^{(2)}x^{(3)}},\nonumber\\
&&\hspace{-1.2cm}S_{12}=0,~~S_{22}=0,~~S_{32}=-\frac{x^{(3)}}{x^{(2)}},~~S_{13}=\frac{(x^{(3)})^3}{18(x^{(2)})^3-3x^{(1)}x^{(2)}x^{(3)}},\nonumber\\
&&\hspace{-1.2cm} S_{23}=\frac{(x^{(3)})^2}{6(x^{(2)})^2-x^{(1)}x^{(3)}},~~S_{33}=\frac{(x^{(3)})^3(-(9x^{(2)})^2+x^{(1)}x^{(3)})}{6(x^{(2)})^3-x^{(1)}x^{(2)}x^{(3)}},\nonumber\\
&&\hspace{-1.2cm}S_{14}=0,~S_{24}=0,~S_{34}=-\frac{4}{3}\frac{x^{(3)}}{x^{(2)}},~S_{15}=0,~S_{25}=0,~S_{35}=-\frac{3x^{(2)}+4tx^{(3)}}{3tx^{(2)}}.\label{4_ex_nul_for}
\end{eqnarray}

We will utilize $S_{3i}$, to establish the integrability of the given equation. The rest of null forms will be used to demonstrate the interconnections. From the null forms $S_{3i},~i=1,2,3,4,5$, with Eq.(\ref{net12}), we can obtain the integrating factors as
\begin{subequations}
\label{4_exa_int}
\begin{eqnarray}
&&\hspace{-0.9cm} R_1=\frac{3(x^{(2)})^4(-2(x^{(2)})^2+x^{(1)}x^{(3)})}{(x^{(3)})^5},~~
 R_2= \frac{x^{(2)}}{(x^{(3)})^2},\\
&&\hspace{-0.9cm} R_3=\frac{3(-6(x^{(2)})^3+x^{(1)}x^{(2)}x^{(3)})}{(x^{(3)})^3},~~R_4=27\frac{(x^{(2)})^4}{(x^{(3)})^4},~~R_5=\frac{t}{3^{\frac{5}{3}}(x^{(2)})^{\frac{4}{3}}}.\label{integral_fourth}
\end{eqnarray}
\end{subequations}
The null forms (\ref{4_ex_nul_for}) and the integrating factors (\ref{4_exa_int}) should satisfy the constraints given in Eqs.(\ref{net13})-(\ref{net16}). 
Since we are dealing a fourth order ODE, the above forms can lead to the four independent integrals. Their forms are given below. 

From the null forms $S_1, S_2$ and $S_3$ we can get the $\lambda$-symmetries by solving the  relation (\ref{lamus}). The resultant expressions read 

\begin{equation}
\lambda_1=\frac{x^{(2)}}{x^{(1)}},~~\lambda_2=0,~~\lambda_3=\frac{1}{t}+\frac{x^{(2)}}{x^{(1)}},~~\lambda_4=\frac{1}{t},~~\lambda_5=\frac{-x+t(x^{(1)}+tx^{(2)})}{t(tx^{(1)}-x)}.
\end{equation}
On the other hand, the null forms $S_{12}, S_{22}$ and $S_{32}$ provide another $\lambda$, that is
\begin{equation}
\lambda_6=\frac{x^{(1)}}{x}.
\end{equation}
Once $\lambda$ is known we can proceed to find the characteristics with the help of (\ref{charn}) which in turn read
\begin{equation}
Q_1=-x^{(1)},~~Q_2=1,~~Q_3=-tx^{(1)},~~Q_4=t,~~Q_5=tx-t^2x^{(1)}~~Q_6=x.
\end{equation}
From these characteristics, we can derive the Lie point symmetries which leave Eq.(\ref{liec}) invariant, that is
\begin{equation}
\hat{v}_1=\frac{\partial}{\partial t},~\hat{v}_2=\frac{\partial}{\partial x},~\hat{v}_3=t\frac{\partial}{\partial t},~~\hat{v}_4=t\frac{\partial}{\partial x},~\hat{v}_5=t^2\frac{\partial}{\partial t}+tx \frac{\partial}{\partial x}~\hat{v}_6=x\frac{\partial}{\partial x}.
\end{equation}

In the following we illustrate the interconnections by considering the vector fields $\hat{v}_1$, $\hat{v}_2$, $\hat{v}_3$ and $\hat{v}_4$.

From the $\lambda$-symmetries we can find the functions $V,~X_1$ and $X_2$ which are of the form (vide Eqs.(\ref{veq}), (\ref{met141}) and (\ref{met14}))
\begin{subequations}
\begin{eqnarray}
&&\hspace{-1cm} V_1=\frac{(x^{(2)})^7(6(x^{(2)})^2-3x^{(1)}x^{(2)}x^{(3)})}{(x^{(3)})^6},~~X_{11}=0,~~X_{21}=\frac{(x^{(2)})^8}{(x^{(3)})^4},\\
&&\hspace{-1cm} V_2=x^{(2)},~~X_{12}=0,~~X_{22}=0,\\
&&\hspace{-1cm} V_3=\frac{6(x^{(2)})^3-x^{(1)}x^{(2)}x^{(3)}}{x^{(3)}},~~X_{13}=\frac{(x^{(3)})^2}{3},~~X_{23}=x^{(2)}x^{(3)},\\
&&\hspace{-1cm} V_4=\frac{(x^{(2)})^4}{x^{(3)}},~~X_{14}=0,~~X_{24}=0.
\end{eqnarray}
\end{subequations}

From the integrating factors and the quantity $V$, we can find the Darboux polynomials associated with the Eq.(\ref{fourth_example}) through the relation $F=\frac{V_i}{R_i},~i=1,2,3,4$. The obtained Darboux polynomials are given by
\begin{equation}
F_1=-\frac{(x^{(2)})^4}{x^{(3)}},~~F_2=F_3=(x^{(3)})^2,~~F_4=\frac{(x^{(3)})^2}{27}.
\end{equation}
Once we know the integrating factors and null forms, we can obtain the component $\zeta^{(1)}$ of the generalized vector field using the Eq.(\ref{gen_vec_ps}). The resultant expressions read
\begin{eqnarray}
&&\zeta_1^{(1)}=\frac{x^{(1)} x^{(2)}}{x (x^{(3)})^2},~~
\zeta_2^{(1)}=\frac{3 (x^{(2)})^5 \left(x^{(1)} x^{(3)}-2 (x^{(2)})^2\right)}{x^{(1)} (x^{(3)})^5}, \nonumber \\
&&\zeta_3^{(1)}=\frac{3 x^{(2)} \left(x^{(1)} x^{(2)} x^{(3)}-6 (x^{(2)})^3\right)}{x^{(1)} (x^{(3)})^3},~~
\zeta_4^{(1)}=\frac{27 (x^{(2)})^5}{x^{(1)} (x^{(3)})^4}.
\end{eqnarray}
The generalized $\lambda$-symmetries are then given by
\begin{eqnarray}
 \gamma_1^{(4)}&=& \frac{3(x^{(2)})^4(-2(x^{(2)})^2+x^{(1)}x^{(3)})}{(x^{(3)})^5} \frac{\partial}{\partial x}+\frac{x^{(1)} x^{(2)}}{x (x^{(3)})^2} \frac{\partial}{\partial x^{(1)}} +\frac{(x^{(2)})^2}{x (x^{(3)})^2} \frac{\partial}{\partial x^{(2)}}\nonumber \\ &&+\frac{x^{(2)}}{x x^{(3)}}\frac{\partial}{\partial x^{(3)}} +\frac{4}{3 x} \frac{\partial}{\partial x^{(4)}},\nonumber \\
 \gamma_2^{(4)}&=&  \frac{x^{(2)}}{(x^{(3)})^2} \frac{\partial}{\partial x} +\frac{3 (x^{(2)})^5 \left(x^{(1)} x^{(3)}-2 (x^{(2)})^2\right)}{x^{(1)} (x^{(3)})^5} \frac{\partial}{\partial x^{(1)}}\nonumber \\&&+\frac{3 \left(x^{(1)} (x^{(2)})^4 x^{(3)}-2 (x^{(2)})^6\right)}{x^{(1)} (x^{(3)})^4}  \frac{\partial}{\partial x^{(2)}} +\frac{4 \left(x^{(1)} (x^{(2)})^3 x^{(3)}-2 (x^{(2)})^5\right)}{x^{(1)} (x^{(3)})^3}  \frac{\partial}{\partial x^{(3)}}\nonumber \\&&+\frac{4 \left(x^{(1)} (x^{(2)})^3 x^{(3)}-2 (x^{(2)})^5\right)}{x^{(1)} (x^{(3)})^3}  \frac{\partial}{\partial x^{(4)}},\nonumber \\
 \gamma_3^{(4)}&=& \frac{3(-6(x^{(2)})^3+x^{(1)}x^{(2)}x^{(3)})}{(x^{(3)})^3} \frac{\partial}{\partial x}  +\frac{3 x^{(2)} \left(x^{(1)} x^{(2)} x^{(3)}-6 (x^{(2)})^3\right)}{x^{(1)} (x^{(3)})^3}\frac{\partial}{\partial x^{(1)}} \nonumber \\&& +\frac{3 \left(x^{(1)} x^{(2)} x^{(3)}-6 (x^{(2)})^3\right)}{x^{(1)} (x^{(3)})^2} \frac{\partial}{\partial x^{(2)}}+(4-\frac{24 (x^{(2)})^2}{x^{(1)} x^{(3)}})  \frac{\partial}{\partial x^{(3)}} \nonumber \\&&+(\frac{20 x^{(3)}}{3 x^{(2)}}-\frac{40 x^{(2)}}{x^{(1)}}) \frac{\partial}{\partial x^{(4)}},\nonumber 
 \end{eqnarray}
 \begin{eqnarray}
 \gamma_4^{(4)}&=&27\frac{(x^{(2)})^4}{(x^{(3)})^4}\frac{\partial}{\partial x}+\frac{27 (x^{(2)})^5}{x^{(1)} (x^{(3)})^4}  \frac{\partial}{\partial x^{(1)}}+\frac{27 (x^{(2)})^4}{x^{(1)} (x^{(3)})^3}  \frac{\partial}{\partial x^{(2)}} \nonumber \\ &&+\frac{27 (x^{(2)})^4}{x^{(1)} (x^{(3)})^3} \frac{\partial}{\partial x^{(3)}}+\frac{60 (x^{(2)})^2}{x^{(1)} x^{(3)}}  \frac{\partial}{\partial x^{(4)}}.
\end{eqnarray}

From the null forms (\ref{4_ex_nul_for}) and the integrating factors (\ref{integral_fourth}) we can find the integrals using the relation (\ref{net17}) admitted by Eq.(\ref{fourth_example}). The integrals are found to be
\begin{eqnarray}
&&I_1=\frac{(x^{(2)})^4(\frac{-3}{2}(x^{(2)})^2+x^{(1)}x^{(3)})}{(x^{(3)})^4},~~
I_2=\frac{t}{3}+\frac{x^{(2)}}{x^{(3)}},\nonumber\\
&&I_3=x-\frac{9(x^{(2)})^3}{(x^{(3)})^2}+\frac{3x^{(1)}x^{(2)}}{x^{(3)}},~~
I_4=9\frac{(x^{(2)})^4}{(x^{(3)})^3}.
\end{eqnarray}
Using the above integrals, we can write the general solution of Eq.(\ref{fourth_example}) as
\begin{equation}
 x(t)=I_3+\frac{3(I_4^2-6I_2(-3 I_1+t)^2)}{6I_1 I_4-2I_4t}.
\end{equation}
Finally the same procedure can be followed for other two vector fields $\hat{v}_5$ and $\hat{v}_6$. However for these two vector fields we find that both the quantifiers turn out to functionally dependent ones. Here also we can start with any one of the methods discussed above, and interconnect every method with all other methods.

\subsection{$n^{th}$-order ODEs}
To verify the validity of our results in the case of $n^{th}$-order ODE we consider the following example, that is
\begin{equation}
 x^{(n)}=\frac{d^n x}{dt^n}=0.\label{nfree}
\end{equation}
Since the Lie point symmetries of Eq.(\ref{nfree}) are explicitly known, we consider Eq.(\ref{nfree}) as the suitable example to verify our results at the $n^{th}$-order. 

In the first column of Table 1, we present the Lie point symmetries of $n^{th}$-order ODE. From the Lie point symmetries we proceed to construct $\lambda$-symmetries, null forms, integrating factors, Darboux polynomials, first prolongation component of generalized $\lambda$-symmetries and finally the integrals for the $n^{th}$-order ODE through the procedure outlined in the previous section. The explicit expressions of all of them are given in Table 1. We have verified all these quantifiers satisfy their respective determining equations. The results support the theory developed in this paper is applicable to $n^{th}$-order ODEs. 

We mention here that the second-order free particle equation admits eight Lie point symmetries. However in the Table 1, one may observe that when $n=2$ (second-order ODE), one can get six Lie point symmetry generators only. The reason is that the six generators which are given in the Table 1 appear in every order starting from second-order and so we have generalized these six symmetry vector fields. The other two vector fields $x\frac{\partial}{\partial t}$ and $xt\frac{\partial}{\partial t}+x^2\frac{\partial}{\partial x}$ of the second-order ODE do not appear in the higher-order. Since these two vector fields are missing in all higher-order we have not included them in the Table 1.

\begin{sidewaystable}[h]
{\scriptsize
\begin{center}
\begin{tabular}{|@{}c@{}|@{}c@{}|@{}c@{}|@{}c@{}|@{}c@{}|@{}c@{}|@{}c@{}|@{}c@{}|@{}c@{}|}\hline
$V$ & $\lambda$ & $S$ & $X$ & $V$ & $(R=\Lambda)$ & $F(=M^{-1})$ & $\zeta^{(1)}$ & $I$\\\hline

$ $ & $ $ & $ $ & $ $ & $ $ & $ $ & $ $ & $ $ & $ $\\
$\frac{\partial}{\partial t}$ & $\frac{x^{(2)}}{x^{(1)}}$ & $S_{ij}=0,j=1$ & $0$ & $1$ & $-1$ & $-1$ & $-\frac{x^{(2)}}{x^{(1)}}$ & $x^{(n-1)}$ \\
$$ & $$ & $i=1,2,...,n-1$ & $$ & $$ & $$ & $$ & $$ & $$\\\hline

$$ & $$ & $$ & $$ & $$ & $$ & $$ & $$ & $ $\\
$\frac{\partial}{\partial x}$ & $0$ & $S_{ij}=0,j=1$ & $0$ & $1$ & $-1$ & $-1$ & $0$ & $x^{(n-1)}$ \\
$$ & $$ & $i=1,2,...,n-1$ & $$ & $$ & $$ & $$ & $$ & $$\\\hline

$$ & $$ & $S_{ij}=0,i=1,2,...,m-1$ & $X_{ij}=0,i=1,2,...,m-1$ & $$ & $$ & $$ & $$ & $ $\\\cline{3-4}
$t^m\frac{\partial}{\partial x}$ & $ $ & $S_{ij}=\frac{(j-1)!}{(-t)^{n-i}(j-(n-i+1))!}$ & $X_{ij}=\frac{(j-1)!(-t)^{j-(n-i+1)}}{(j-(n-i+1))!}$ & $$ & $$ & $$ & $$ & $ $\\
$$ & $$ & $i=n-1,n-2,...,m$ & $i=n-2,n-3,...,m$ & $$ & $$ & $$ & $$ & $$\\\cline{3-4}
$$ & $\frac{m}{t}$ & $j=n-(m-1),m\neq1$ & $j=n-(m-1),m\leq n-2,m\neq1$ & $t^{n-m}$ & $\frac{(-1)^jt^{n-m}}{(j-1)!}$ & $\frac{(j-1)!}{(-1)^j}$ & $m\frac{(-1)^jt^{n-m-1}}{(j-1)!}$ & $\displaystyle\sum_{m=n-j}^{n-1}\frac{x^{(m)(-t)^{m-(n-j)}}}{(m-(n-j))!}$\\\cline{3-4}
$m=n-1,$ & $ $ & $S_{ij}=\frac{(n-1)!}{(-t)^{\delta}(n-(\delta+1))!}$ & $X_{ij}=\frac{(-1)^{\alpha}(n-1)!t^{n-(\alpha+1)}}{(n-(\alpha+1))!}$ & $$ & $$ & $$ & $$ & $ $\\
$n-2,...,1$ & $ $ & $i=n-\delta,j=n,m=1$ & $i=n-\alpha,j=\alpha+1,\alpha+2,...,n$ & $$ & $$ & $$ & $$ & $$\\
$$ & $ $ & $\delta=1,2,3,...,(n-1)$ & $\alpha=2,3,4,...,(n-1),m=1$ & $$ & $$ & $$ & $$ & $$\\\hline

$$ & $ $ & $S_{ij}=\frac{(n-1)!}{(-t)^{\delta}(n-(\delta+1))!}$ & $X_{ij}=\frac{(-1)^{\alpha}(n-1)!t^{n-(\alpha+1)}}{(n-(\alpha+1))!}$ & $$ & $$ & $$ & $$ & $ $\\
$t\frac{\partial}{\partial t}$ & $\frac{x^{(2)}}{x^{(1)}}+\frac{1}{t}$ & $i=n-\delta,~j=n$ & $i=n-\alpha,j=\alpha+1,\alpha+2,...,n$ & $t^{n-1}$ & $\frac{(-1)^nt^{n-1}}{(n-1)!}$ & $\frac{(n-1)!}{(-1)^n}$ & $(\frac{x^{(2)}}{x^{(1)}}+\frac{1}{t})\frac{(-1)^nt^{n-1}}{(n-1)!}$ & $\displaystyle\sum_{m=n-j}^{n-1}\frac{x^{(m)(-t)^{m-(n-j)}}}{(m-(n-j))!}$\\
$$ & $$ & $\delta=1,2,3,...,(n-1)$ & $\alpha=2,3,4,...,(n-1)$ & $$ & $$ & $$ & $$ & $$\\\hline

$$ & $ $ & $S_{ij}=0$ & $$ & $$ & $$ & $$ & $$ & $$\\
$x\frac{\partial}{\partial x}$ & $\frac{x^{(1)}}{x}$ & $S_{(n-1)j}=\frac{-x^{(n-1)}}{x^{(n-2)}}$ & $0$ & $x^{(n-2)}$ & $\frac{x^{(n-2)}}{(x^{(n-1)})^2}$ & $(x^{(n-1)})^2$ & $\frac{x^{(1)}x^{(n-2)}}{x(x^{(n-1)})^2}$ & $\hat{I}=-t+\frac{x^{(n-2)}}{x^{(n-1)}}$\\
$$ & $$ & $i=1,2,...,n-2,j=n+1$ & $$ & $$ & $$ & $$ & $$ & $$\\\hline

$$ & $$ & $S_{ij}=\frac{(n-1)!}{(-t)^{\delta}(n-(\delta+1))!}$ & $X_{ij}=\frac{(-1)^{\alpha}(n-1)!t^{n-(\alpha+1)}}{(n-(\alpha+1))!}$ & $$ & $$ & $$ & $$ & $ $\\
$\frac{(n-1)}{2}xt\frac{\partial}{\partial x}$ & $a^{*}$ & $i=n-\delta,~j=n$ & $i=n-\alpha,j=\alpha+1,\alpha+2,...,n$ & $t^{n-1}$ & $\frac{(-1)^nt^{n-1}}{(n-1)!}$ & $\frac{(n-1)!}{(-1)^n}$ & $a^{*}\frac{(-1)^nt^{n-1}}{(n-1)!}$ & $\displaystyle\sum_{m=n-j}^{n-1}\frac{x^{(m)(-t)^{m-(n-j)}}}{(m-(n-j))!}$\\
$+\frac{t^2}{2}\frac{\partial}{\partial t}$ & $ $ & $\delta=1,2,3,...,(n-1)$ & $\alpha=2,3,4,...,(n-1)$ & $$ & $$ & $$ & $$ & $ $\\\hline
\end{tabular}\\
$a^{*}$ is given by $\frac{(n-1)x-x^{(2)}t^2+(n-3)tx^{(1)}}{(n-1)xt-x^{(1)}t^2}$
\caption{Interconnections among various quantities admitted by the linear ODE $x^{(n)}=0$, $n \geq 3$. For $n=2$, see text.}
\end{center}}
\end{sidewaystable}

\section{Conclusion}
In this work, we have developed a systematic procedure to make interconnections among several analytical methods which are available in the literature for $n^{th}$ order nonlinear ODEs. In the Prelle-Singer procedure, the important quantities are the null forms $S_i,~i=1,2,3,...,n-1$ and the integrating factor $R$.  By introducing suitable transformations in the null forms $S_i$ and the integrating factor $R$, we can relate the other methods such as $\lambda$-symmetries, Lie point symmetries, Jacobi last multiplier, Darboux polynomials, adjoint-symmetries method and generalized $\lambda$-symmetries with the Prelle-Singer procedure. The $\lambda$-symmetries, Lie point symmetries and generalized $\lambda$-symmetries are connected with the Prelle-Singer procedure through null forms $S_i$. The other methods, Jacobi last multiplier, Darboux polynomials and the adjoint-symmetries are related with the Prelle-Singer method by the integrating factors and the expression $V$ which can be obtained from the null form $S_{n-1}$. By utilizing these interconnections, one can find the relevant quantities associated with the other methods without solving their own determining equations. We have explained the above said interconnections with several examples, that is, second-order, third-order, fourth-order and $n^{th}$-order ODEs. 

Finally, we would to like to point out that the interconnections which we have presented in the manuscript will work for partially integrable systems also (that is lesser number of integrals than the required $n$ independent integrals). For example, if we know one integral, we can obtain the quantities associated with the other methods corresponding to this integral using the interconnections which we have discussed above.

Also, there may be situations where the above interconnections are difficult to acheive. For example, consider the second-order ODE $\ddot{x}+\frac{t^2}{4 x^3}+x+\frac{1}{2 x}=0$ which has no Lie point symmetries but has a $\lambda$-symmetry vector field $\tilde{V}=x\frac{\partial}{\partial x}$ with the $\lambda$-function $\lambda=\frac{t}{x^2}$ \cite{Mur1}. For this $\lambda$-function, we can have the expression $\tilde{V}^{[\lambda,(1)]}I=x\frac{\partial I}{\partial x}+(\frac{t}{x}+\dot{x})\frac{\partial I}{\partial \dot{x}}=0$. Replacing the latter expression with the Prelle-Singer method expressions $\frac{\partial I}{\partial x}=-RS$ and $\frac{\partial I}{\partial \dot{x}}=-R$, and solving it, we can get the null form $S$ as $S=-\frac{t+x \dot{x}}{x^2}$. The corresponding integrating factor $R$ can be found as $R=\frac{1}{x \left(\frac{(t+2 x \dot{x})^2}{4 x^4}+1\right)}$. So we can find the rest of the integrability quantifiers. But in the case where $\xi \neq 0$, it is very difficult to solve the expression $\tilde{V}^{[\lambda,(1)]}I=\xi \frac{\partial I}{\partial t} +\eta \frac{\partial I}{\partial x}+\eta^{(1)}\frac{\partial I}{\partial \dot{x}}=0$. If it is possible to solve the latter expression, one can find the other integrability quantifiers appropriately. Otherwise the interconnections may not work fully. In the same way, we can obtain more number of Darboux polynomials by making combinations of them. Doing so, we can obtain the Darboux polynomials with cofactor apart from $\phi_{\dot{x}}$. So the interconnections will work for arbitrary $\lambda$-symmetries, Jacobi last multipliers or Darboux polynomials provided we are able to solve the corresponding equations
\vspace{0.18cm}\\
\begin{small}
\textit{Ethics statement}. This research did not involve human or animal subjects.\\
\textit{Data accessibility}. This paper does not have any data.\\
\textit{Competing interests statement}. We have no competing interests.\\
\textit{Authors contributions statement}. All the authors have contributed equally to the research and to the writing up of the paper.\\
\textit{Funding statements}. RMS acknowledges the University Grants Commission (UGC-RFSMS), Government of India, for providing a Research Fellowship. The work of VKC is supported by INSA young scientist project. The work of MS forms part of a research project sponsored by Department of Science and Technology, Government of India.  The work of ML is supported by a Department of Science and Technology (DST), Government of India, IRHPA research project. ML is also supported by a DAE Raja Ramanna Fellowship and a DST Ramanna Fellowship programme.
\end{small}

\end{document}